\documentclass{article}

\usepackage[utf8]{inputenc}

\usepackage{amsmath,amssymb}

\usepackage[english]{babel}

\usepackage{caption}
\DeclareCaptionLabelFormat{cont}{#1~#2\alph{ContinuedFloat}}

\makeatletter
\newcommand\resetsubfigs{\setcounter{sub\@captype}{0}}
\makeatother

\usepackage{color}
\definecolor{purple}{rgb}{0.5,0,0.5}

\usepackage{amsmath,amssymb,amsfonts,amsthm}
\usepackage{physics}
\usepackage{dsfont}
\usepackage{bbm}
\usepackage{graphicx}

\usepackage{tikz}
\usepackage{tikz-cd}
\usepackage{subcaption}

\usepackage{comment}

\usetikzlibrary{arrows,shapes,trees}

\newcommand{\ra}{\rightarrow}

\newcommand{\End}{{\rm End}}

\newcommand{\CC}{{\mathbb C}}
\newcommand{\ZZ}{{\mathbb Z}}
\newcommand{\RR}{{\mathbb R}}

\newcommand{\HH}{{\mathbb H}}

\newcommand{\cO}{\mathcal O}

\newcommand{\cG}{\mathcal G}

\newcommand{\Q}{{\mathcal Q}}

\newcommand{\fotimes}{\mathbin{\widehat\otimes}}

\usepackage{mathtools}
\usepackage{xfrac}

\newtheorem{theorem}{Theorem}[section]
\newtheorem{lemma}[theorem]{Lemma}

\usepackage[left=1in,right=1in,top=1in,bottom=1in]{geometry}

\usepackage{dcolumn}
\usepackage{bm}

\usepackage[capitalize]{cleveref}

\title{Fermionic Matrix Product States and One-Dimensional Short-Range Entangled Phases with Anti-Unitary Symmetries}
\author{Alex Turzillo and Minyoung You \\ {\small\it California Institute of Technology, 1200 E California Blvd, Pasadena, CA 91125}}
\date{September 29, 2017}

\begin{document}

\maketitle




\begin{abstract}
We extend the formalism of Matrix Product States (MPS) to describe one-dimensional gapped systems of fermions with both unitary and anti-unitary symmetries. Additionally, systems with orientation-reversing spatial symmetries are considered. The short-ranged entangled phases of such systems are classified by three invariants, which characterize the projective action of the symmetry on edge states. We give interpretations of these invariants as properties of states on the closed chain. The relationship between fermionic MPS systems at an RG fixed point and equivariant algebras is exploited to derive a group law for the stacking of fermionic phases protected by general fermionic symmetry groups.


\end{abstract}

\section{Introduction and Overview}

Matrix Product States (MPS) have proven useful at describing the ground states of gapped local Hamiltonians in one spatial dimension \cite{Hastings,MPSreview}. This representation leads to a classification of interacting short-range entangled (SRE) bosonic phases with a symmetry $G$ in terms of the group cohomology of $G$ \cite{ChenGuWenone,fidkit}. One-dimensional systems of fermions are related to these bosonic systems by the Jordan-Wigner transformation, and this fact has been exploited to classify fermionic SRE phases \cite{fidkit,ChenGuWentwo}. The language of MPS was recently extended to an intrinsically fermionic formalism in what has been dubbed fermionic MPS (fMPS) \cite{FMPS,KTY2}. This approach has the benefits of straightforwardly describing systems on a closed chain with twisted boundary conditions and allowing one to derive a group law for the stacking of fermionic SRE phases. It also makes manifest the relation to spin-TQFT, which is conjectured to describe the long-distance physics of fermionic systems \cite{KTY2}. Bosonic and fermionic systems with a time-reversal symmetry have also been classified by use of MPS, by allowing the projective action of the symmetry on edge degrees of freedom to be anti-unitary \cite{fidkit,ChenGuWentwo}. It is natural to ask whether the formalism of fMPS can be applied to these systems. Bosonic systems with anti-unitary symmetries have previously been studied with an approach similar to ours \cite{ShiozakiRyu}. Fermionic systems have been described in Ref. \cite{FMPS}, and their results agree with ours.

The structure of the paper and its main results are as follows. In Section \ref{fermionicMPS}, we review the formalism of $\cG$-equivariant fMPS for a unitary on-site symmetry $\cG$. We recall how fermionic SRE phases are classified by Morita classes of equivariant algebras and how the invariants $\alpha$, $\beta$, and $\gamma$ that characterize these algebras appear in the action of $\cG$ on edge degrees of freedom. We then derive interpretations of the invariants on the closed chain that extend the results of Ref. \cite{KapustinThorngren}, which were discovered in the context of spin-TQFT. Next, time-reversing symmetries and their relation to spatial parity are discussed. The generalizations of the three invariants to phases with such symmetries are derived and interpreted. In Section \ref{stackingsection}, a general stacking law \eqref{stackinglaw} is derived for fermionic SRE phases with a symmetry $\cG$ that is a central extension by fermion parity of a bosonic symmetry group that may contain anti-unitary symmetries, and we contrast this result with the bosonic stacking law. In Section \ref{examplesection}, we consider several examples, recovering the $\ZZ/8$ classification of fermionic SRE phases of symmetry class BDI $(T^2=1)$ and the $\ZZ/2$ classification of class DIII $(T^2=P)$.

The authors would like to thank Anton Kapustin for helpful discussions throughout the production of this paper. This research was supported by the U.S. Department of Energy, Office of Science, Office of High Energy Physics, under Award Number DE-SC0011632.

\section{Fermionic MPS}\label{fermionicMPS}

\subsection{Unitary symmetries}

We begin by briefly recalling the fMPS formalism, leaving many of the details to Ref. \cite{KTY2}.

The symmetries of a fermionic system form a finite supergroup $(\cG,p)$; that is, a finite group $\cG$ with a distinguished central involution $p\in \cG$ called \emph{fermion parity}. Every supergroup arises as a central extension of a group $G_b\simeq\cG/\ZZ_2$ of bosonic symmetries by $\ZZ/2=\{1,p\}$:
\begin{equation}
\ZZ_2\xrightleftharpoons[t]{i}\cG\xrightleftharpoons[s]{b}G_b.
\end{equation}
Such extensions are classified by cohomology classes $[\rho]\in H^2(G_b,\ZZ/2)$. A trivialization $t:\cG\ra\ZZ/2$ with $t(p)=0$ defines a representative $\rho$ of this class by $\rho(b(g),b(h))=(\delta t)(g,h)$.

MPS is an ansatz for constructing one-dimensional gapped systems with such a symmetry. A translation-invariant $\cG$-symmetric MPS system consists of the following data: a physical one-site Hilbert space $A$ with a unitary action $R$ of $\cG$, a virtual space $V$ with a projective action $Q$ of $\cG$, and an injective MPS tensor $T: A \rightarrow \End V$ that satisfies an equivariance condition: for all $a \in A,g\in\cG$,
\begin{equation}\label{equivariance}T(R(g) a) = Q(g) T(a) Q(g)^{-1}.\end{equation}From this data, one may construct a gapped, symmetric, and frustration-free lattice Hamiltonian \cite{MPSclassif}.

For our purposes, we are only interested in MPS at an RG fixed point, where a system exhibits the physics universal to its gapped phase. Physical sites are blocked together under real-space RG. At a fixed point, this procedure endows the space $A$ with a product $m:A\otimes A\ra A$, making $A$ into a finite-dimensional associative algebra and the space $V$ into a faithful module over $A$ with structure tensor $T$ \cite{KTY1}; that is,\begin{equation}T(a)T(b)=T(m(a\otimes b))\end{equation}for all $a,b\in A$. It follows that $T(m(R(g)a \otimes R(g) b))=T(R(g) m(a\otimes b))$. Since $T$ is injective, $m$ satisfies\begin{equation}\label{equialg}R(g)m(a\otimes b) = m(R(g)a \otimes R(g) b).\end{equation}We say $A$ is a $\cG$-equivariant algebra. Because $T$ can be put into a canonical form, $A$ is semisimple \cite{KTY1}. It can also be viewed as the algebra of linear operators on the space of low energy boundary states \cite{fidkit}.

Our approach is to work with ground states of MPS systems, rather than their Hamiltonians. Every ground state of an MPS Hamiltonian has the form of a generalized MPS, which we now describe. Given a $(\cG, p)$-equivariant MPS system described by an algebra $A$ and tensor $T: A \rightarrow \End (V)$, one obtains a generalized MPS by choosing an observable $X\in\End(V)$ that \emph{supercommutes} with $T(a)$;\footnote{Herein lies the difference between the bosonic and fermionic MPS formalisms. In the bosonic case, $X$ commutes with $T(a)$ regardless of its charge under $P$ (which is not a distinguished symmetry). For this reason, the twisted sector state spaces and how they are acted on by symmetries typically differ between a fermionic system and its Jordan-Wigner transform.} that is
\begin{equation}
\label{obs}XT(a)=(-1)^{|a||X|}T(a)X
\end{equation}
for all $a\in A$ such that $R(p)a=(-1)^{|a|}a$, where the parity of $X$ is defined by $PXP^{-1}=(-1)^{|X|}X$, for $P:=Q(p)$. A linear map $X$ satisfying this condition is called an even or odd $\ZZ/2$-graded module endomorphism, depending on its parity.

For $X$ satisfying the supercommutation rule \eqref{obs}, the generalized MPS has conjugate wavefunction
\begin{equation}\label{twistedMPS}
\bra{\psi_{T,g}^X}=\sum_{i_1,\ldots,i_N}\Tr_U[Q(g)XT(e_{i_1})\cdots T(e_{i_N})]\bra{i_1\cdots i_N}
\end{equation}
in the $g$-twisted sector. When $\cG$ splits as $G_b\times\{1,p\}$, one can choose a splitting $t$ in order to make sense of the $g$-twisted sector as a $b(g)$-twisted NS or R sector. The $g$-twisted sector for $t(g)=0$ consists of states on the circle with NS spin structure and a $G_b$ gauge field of holonomy $b(g)$, while the $g$-twisted sector for $t(g)=1$ consists of states on the circle with R spin structure and a $G_b$ gauge field of holonomy $b(g)$.  When $\cG$ is non-split, one does not have independent spin structures and gauge fields.\footnote{The appropriate geometric structure -- the $\cG$-Spin structure -- is discussed in Ref. \cite{KTY2}.}

The MPS description of a gapped system is not unique. For this reason, fixed point systems -- and therefore gapped phases -- are classified not by the algebras themselves but by their Morita classes \cite{KTY1}. After all, it is the module category (the $T$'s and $X$'s) of $A$ that determines the system's ground states.

\subsection{Equivariant algebras for fermionic SRE phases}

For the remainder of the paper, we will restrict our attention to SRE phases. The algebras that correspond to fixed points in such phases are either of the form $\End(U)$ for $U$ a projective representation $Q$ of $\cG$ or of the form $\End(U_b) \otimes\CC\ell(1)$, for $U_b$ a projective representation of $G_b$, with the following actions of $\cG$ \cite{KTY2}. We refer to algebras of the first (second) type as ``even'' (``odd''). Odd algebras are only present when the extension for $\cG$ splits.

The action of $\cG$ on an even algebra $A=\End(U)$ is simply
\begin{equation}
R(g)\cdot M=Q(g)MQ(g)^{-1}.
\end{equation}
Two even algebras are Morita equivalent if their projective representations have the same $[\omega]\in H^2(\cG, U(1))$ \cite{Ostrik}. It is shown in the appendix that $[\omega]$ is equivalent to a pair $(\alpha, \beta) \in C^2(G_b, U(1)) \times C^1(G_b, \ZZ/2)$ that satisfies $\delta \alpha = \sfrac{1}{2} \beta \cup \rho$ and $\delta \beta = 0$, up to coboundaries.\footnote{In our notation, $\alpha$ takes values in $\mathbb{R}/\ZZ = [0, 1)$ and $\sfrac{1}{2} \beta \cup \rho$ in $\{0, \sfrac{1}{2} \} \subset [0,1)$. Crucially, $\alpha$ is defined up to a $\cG$-coboundary with arguments in $G_b$; only when $\cG$ splits is this a $G_b$-coboundary. This subtlety is relevant when $\cG=\ZZ_4^{FT}$, see Section \ref{examplesection}.} In particular, when $\cG$ splits, $\rho$ is trivial and the equivalence classes defining the phase are $[\alpha],[\beta]\in H^2(G_b,U(1))\times H^1(G_b,\ZZ/2)$. The invariant $[\alpha]$ is simply $[\omega]$ pulled back by $s$ to $G_b$, while $\beta$ measures the parity of $Q(g)$:\begin{equation}\label{Qgparity}PQ(g)P^{-1}=e^{i\pi\beta(b(g))}Q(g).\end{equation}

Let $\Gamma$ denote the generator of $\CC\ell(1)$ with $\Gamma^2=+1$. The action of $\cG$ on an odd algebra is
\begin{equation}\label{oddaction}
R(g)\cdot M\otimes\Gamma^m=(-1)^{(\beta(g)+t(g))m}Q(g)MQ(g)^{-1}\otimes\Gamma^m
\end{equation}
Morita classes of odd algebras are classified by the class $[\alpha] \in H^2(G_b, U(1))$ that describes the projective $G_b$-action $Q$ and the class $ [\beta] \in H^1(G_b, \ZZ/2)$ that describes the action of $G_b$ on $\Gamma$.

It is apparent from \eqref{oddaction} that changing the trivialization $t$ shifts $\beta$ by an arbitrary $\mu=(t'-t)\in Z^1(G_b;\ZZ/2)$. Therefore, unless $t$ is fixed, $\beta$ is not a well-defined invariant of odd algebras. On the other hand, for even algebras, redefining the action of $g_b$ to $Q(g_b)P^{\mu(g_b)}$ leaves $\beta$ unchanged but shifts $\alpha$ by $\sfrac{1}{2}\beta\cup\mu$.

In summary, when $\cG$ splits as $G_b \times \ZZ_2^F$, the $\cG$-symmetric fermionic SRE phases are classified by $[\alpha],[\beta],[\gamma] \in H^2(G_b, U(1)) \times H^1(G_b, \ZZ/2) \times \ZZ/2$ where $\gamma \in \ZZ/2$ tells us whether the algebra is even ($\gamma = 0$) or odd ($\gamma = 1$). When $\cG$ does not split, only the invariants $[\alpha]$ (not a cocycle) and $[\beta]$ are present.

A system in an SRE phase has exactly one state per twisted sector. To see this from the algebra, count independent solutions $X$. The unique simple module over an even algebra $\End(U)$ is $U$, and by Schur's lemma its only endomorphism is $X=\mathds{1}$. The unique faithful simple module over an odd algebra $\End(U_b)\otimes\CC\ell(1)$ is $U_b\otimes\CC^{1|1}$. It has two endomorphisms - an even one $X=\mathds{1}\otimes\mathds{1}$ and an odd one $X=\mathds{1}\otimes\sigma_y$. The former appears in the wavefunction of the NS sector MPS state and the latter for the R sector MPS state.

\subsection{Invariants of fermionic SRE phases}

The invariants $\alpha$, $\beta$, and $\gamma$ can also be extracted from an SRE fermionic MPS system without reference to the algebra $A$. Below we give a physical interpretation of these invariants as observable quantities.

We begin by studying how the MPS in the $g$-twisted sector transforms under the action of a unitary symmetry $h\in\cG_0$. Let $\omega$ be the cocycle that characterizes the projective action $Q$ on the module. Then \begin{align}\begin{split}\label{MPStrans1}
R(h)\cdot\Tr[Q(g)XT^i]\bra{i}
&=\Tr[Q(g)XQ(h)^{-1}T^iQ(h)]\bra{i}\\
&=e^{2\pi i(\omega(h,g)+\omega(hg,h^{-1})-\omega(h,h^{-1}))}\Tr[Q(hgh^{-1})[Q(h)XQ(h)^{-1}]T^i]\bra{i}.\end{split}
\end{align}We have used the fact that
\begin{equation}
\omega(h,h^{-1})=\omega(h^{-1},h),
\end{equation}
which follows from the cocycle condition.\footnote{We always work in a gauge $Q(1)=\mathds{1}$.} We see that under the action of a unitary symmetry $h$,\begin{enumerate}
\item The $g$-twisted sector maps to the $hgh^{-1}$-twisted sector.
\item The operator $X$ is conjugated by $Q(h)$.
\item States also pick up a phase of\begin{equation}\label{unitaryphase}e^{2\pi i(\omega(h,g)+\omega(hg,h^{-1})-\omega(h,h^{-1}))}.\end{equation}
\end{enumerate}We are now ready to interpret the three invariants.

\noindent\emph{Gamma.}

Suppose $h=p$ and $g\in\{1,p\}$. Then the phase \eqref{unitaryphase} vanishes, but there is still a sign coming from the conjugation of $X$ by $P$. It is always $+1$ if the algebra is of the from $\End(U)$ (i.e. if $\gamma=0$). If the algebra is of the form $\End(U_b)\otimes\CC\ell(1)$ (i.e. if $\gamma=1$), this sign is $+1$ in the NS sector and $-1$ in the R sector. Therefore we can conclude that the invariant $(-1)^\gamma$ is detected as the fermion parity ($p$-charge) of the $R$ sector state.

\noindent\emph{Beta.}

Continuing to take $h=p$, in the $g$-twisted sector the phase \eqref{unitaryphase} becomes\begin{equation}\label{beta}\sfrac{1}{2}\beta(g):=\omega(p,g)-\omega(g,p).\end{equation}This term satisfies $\beta(pg)=\beta(g)$ and takes values in $\{0,\sfrac{1}{2}\}$; in fact, it defines a $\ZZ/2$-valued cocycle of $G_b$. See the appendix for a proof. When $\gamma=0$, the sign $(-1)^{\beta(g_b)}$ is the fermion parity of the $g$-twisted sector for $g$ with $b(g)=g_b$. If $\cG$ splits, one can equivalently say that $(-1)^{\beta(g)}$ is the parity of the $b(g)$-twisted NS and R sectors. If $\cG$ splits, it is possible that $\gamma=1$. In this case, one must choose a splitting to interpret $\beta$ this way. Then $(-1)^{\beta(g)}$ is still the parity of the $b(g)$-twisted NS sector, but the parity of the $b(g)$-twisted R sector receives a contribution of $-1$ from conjugation of $X$ by $P$, in addition to the $\beta(g)$ term.

Note that $\beta(g)$ also describes the $g$-charge of the $p$-twisted (Ramond) sector for systems with $\gamma=0$. This is no coincidence: the phase \eqref{unitaryphase} agrees with Equation 4.11 of Ref. \cite{KT}, where it was derived from bosonic (i.e. $X=\mathds{1}$) TQFT. If $g$ and $h$ commute, one can sew together the ends of the cylinder to create a torus with holonomies $g$ and $h$ around its cycles. This torus evaluates to the phase
\begin{equation}
\omega(h,g)+\omega(hg,h^{-1})-\omega(h,h^{-1})=\omega(h,g)-\omega(g,h).
\end{equation}
This surface can also be evaluated as a torus with holonomies $h$ and $g^{-1}$, respectively, yielding
\begin{align}\begin{split}
\omega(g^{-1},h)+&\omega(g^{-1}h,g)-\omega(g^{-1},g)=\omega(h,g)+\omega(g^{-1},hg)-\omega(g^{-1},g)=\omega(h,g)-\omega(g,h)
\end{split}\end{align}
These are equal, as is required by consistency of the TQFT. In terms of states, the $h$-charge of the $g$-twisted sector is the same as the $g^{-1}$-charge of the $h$-twisted sector, as long as $g$ and $h$ commute. There is no analogous statement for systems with $\gamma=1$. Recall that $\beta(g)$ measures whether or not $g$ acts as $\sigma_z$ on the second factor of $\End(U)\otimes\CC\ell(1)$. Then $Q(g)$ anticommutes with $X=\mathds{1}\otimes\sigma_y$, and so the state picks up an extra charge of $\beta(g)$ which cancels with the sign \eqref{unitaryphase} for a total $g$-charge of $+1$ in the R sector.

\noindent\emph{Alpha.}

Consider the MPS state on a circle with two adjacent domain walls, parametrized by bosonic symmetries $g_b,h_b\in G_b$, as in Figure \ref{fig:alpha}. Upon fusing them, the state picks up a phase:
\begin{align}\begin{split}
&\Tr[Q(s(g_b))Q(s(h_b))T^i]\bra{i}=e^{2\pi i\omega(s(g_b),s(h_b))}\Tr[Q(s(g_b)s(h_b))T^i]\bra{i}
\end{split}\end{align}
These phases define a $G_b$-cochain
\begin{equation}\label{alphadef}
\alpha(g_b,h_b)=\omega(s(g_b),s(h_b)).
\end{equation}
If $\cG$ splits, then the fact that $\omega$ is a cocycle implies that $\alpha$ is as well. If the extension $\cG$ is instead defined by a nontrivial $\rho$, then $\alpha$ has coboundary $\sfrac{1}{2} \beta\cup\rho$. See the appendix for details. Redefining each $X=\mathds{1}$ by a sector-dependent phase shifts $\alpha$ by a $\cG$-coboundary with arguments in $G_b$, as expected.

Note that when $\beta$ and $\gamma$ are trivial, there are no fermionic states and the system is insensitive to spin structure. In this sense, $\alpha$ captures purely bosonic features of the system.

\noindent\hspace{0.02\textwidth}\fbox{
\begin{minipage}{0.94\textwidth}
\vspace{1mm}
\noindent\emph{In summary.}

\begin{itemize}
\item $(-1)^\gamma$ is the fermion parity of the untwisted R sector.
\item If $\gamma=0$, $(-1)^{\beta(g_b)}$ is the fermion parity of the $g$-twisted sector for either of the two $g$'s with $b(g)=g_b$. Alternatively, $(-1)^{\beta(g_b)}$ is the $g$-charge of the untwisted R sector. If $\gamma=1$, $(-1)^{\beta(g_b)}$ is the fermion parity of the $g_b$-twisted NS sector, as determined by the choice of splitting.
\item $e^{2\pi i\alpha(g_b,h_b)}$ is the phase due to fusing $g_b$ and $h_b$ domain walls.
\end{itemize}
\vspace{-0.1cm}
\end{minipage}
}

\begin{figure}
\centering
\begin{minipage}{.5\textwidth}
  \centering
  
  \begin{tikzpicture}

\draw (0,0) ellipse (1.5 and 0.5);
\draw (0,2.2) ellipse (1.5 and 0.5);
\fill [white] (-2,0) rectangle (2,0.6);
\draw (-1.5,0) -- (-1.5,2.2);
\draw (1.5,0) -- (1.5,2.2);

\draw[thick] (0,1.7) -- (0,0.75);
\draw[thick] (-0.52,-0.47) -- (-0.52,-0.1);
\draw[thick] (0.52,-0.47) -- (0.52,-0.1);
\draw[thick] (-0.52,-0.1) arc (180:120:1);
\draw[thick] (0.52,-0.1) arc (0:60:1);

\fill [black] (0,0.75) circle (0.1);

\node at (-0.32,-0.25) {$h$};
\node at (0.72,-0.25) {$g$};
\node at (0.25,1.45) {$gh$};

\end{tikzpicture}
  
  \captionof{figure}{fusion of domain walls}
  \label{fig:alpha}

\end{minipage}%
\begin{minipage}{.5\textwidth}
  \centering
  
\begin{tikzcd}
 & \ZZ_2^F \arrow{d}{} & \\ & \cG \arrow{d}{b} & \\ G_0 \arrow{r}{} & G_b \arrow{r}{x} & \ZZ_2^T
\end{tikzcd}

\captionof{figure}{symmetry data}
\label{fig:sym}
  
\end{minipage}
\end{figure}

\subsection{Anti-unitary and orientation-reversing symmetries}

More generally, a fermionic system may be invariant under anti-unitary symmetries as well as unitary ones. In this case, the full symmetry group $\cG$ is a central extension by $\ZZ_2^F$ of a bosonic symmetry group $G_b$, which is itself an extension of $\ZZ_2^T$ by a finite group $G_0$, as in Figure \ref{fig:sym}. The symmetry class $(\cG,p,x)$ is determined by a central $p\in\cG$ and a map $x:G_b\ra\ZZ/2$ that encodes whether a bosonic symmetry is unitary or anti-unitary. Note that the composition $x\circ b$, which we also call `$x$,' satisfies $x(p)=0$. Let $\cG_0$ denote its kernel.

A fixed point MPS system of symmetry class $(\cG,p,x)$ consists of a finite-dimensional semisimple associative algebra $A$ and a faithful module $T:A\ra\End(V)$, satisfying the equivariance conditions \eqref{equialg}, \eqref{equivariance} as before, only now the group action may be anti-unitary. In particular, the projective action on $V$ is given by a unitary operator $Q(g)$ for each $g\in\cG_0$ and an anti-unitary operator $Q(g)$ for each $g\notin\cG_0$ that satisfy
\begin{equation}\label{proj1}
Q(g)Q(h)=e^{2\pi i\omega(g,h)}Q(gh)
\end{equation}
for phases $\omega(g,h)$. By comparing $[Q(g)Q(h)]Q(k)$ and $Q(g)[Q(h)Q(k)]$, we find the $x$-twisted cocycle condition:
\begin{equation}\label{twistedcocycle}
\omega(g,h)+\omega(gh,k)=(-1)^{x(g)}\omega(h,k)+\omega(g,hk).\end{equation}
Redefining each $Q(g)$ by a $g$-dependent phase corresponds to shifting $\omega$ by an $x$-twisted coboundary. Therefore the action of $\cG$ on the module $V$ is characterized by a twisted cohomology class $[\omega]\in H^2(\cG,U(1)_T)$. The group action $R$ on $A$ is defined via \eqref{equivariance}. It will be convenient to define linear maps $M(g)$ by
\begin{equation}
M(g)=\left\{\begin{array}{lr}Q(g)&g\in\cG_0\\Q(g)K&g\notin\cG_0\end{array}\right.
\end{equation}
where $K$ denotes complex conjugation.

Unitary symmetries that reverse the orientation of one-dimensional space can also be described in this language. Let $x$ measure whether a symmetry reverses orientation. The natural generalization of \eqref{equivariance} is
\begin{align}\begin{split}\label{Tequimod}T(R(g)a)=M(g)T(a)M(g)^{-1}\qquad &\text{for }g\in\cG_0\\T(R(g)a)=M(g)T(a)^TM(g)^{-1}\qquad &\text{for }g\notin\cG_0.\end{split}\end{align}Let us introduce the following shorthand. For a matrix $\cO\in\End(V)$, define
\begin{equation}\begin{array}{ll}
\cO^{T0}=\cO,&\qquad\cO^{T1}=\cO^T,\\ \{\cO\}^0=\cO,&\qquad\{\cO\}^1=[\cO^{-1}]^T.
\end{array}\end{equation}
Since $R$ is a group homomorphism,
\begin{align}\begin{split}
M(g)\{M(h)\}^{x(g)}T(a)^{Tx(gh)}M(h)^{Tx(g)}M(g)^{-1}&=T(R(g)R(h)a)\\&=T(R(gh)a)\\&=M(gh)T(a)^{Tx(gh)}M(gh)^{-1}.
\end{split}\end{align}
This implies there exists a number $\omega(g,h)\in\RR/\ZZ$ such that
\begin{equation}\label{proj}M(g)\{M(h)\}^{x(g)}=e^{2\pi i\omega(g,h)}M(gh).\end{equation}By comparing the two equal expressions $M(g)\{M(h)\}^{x(g)}\{M(k)\}^{x(gh)}$ and $M(g)\{M(h)M(k)^{x(h)}\}^{x(g)}$, one recovers the $x$-twisted cocycle condition \eqref{twistedcocycle} for $\omega$.

From the perspective of two-dimensional spacetime, it is not surprising that time-reversal\footnote{By a well-known result of Wigner, an anti-unitary symmetry reverses the direction of time.} and space-reversal should be treated similarly. To make the connection more explicit, note that the physical Hilbert space carries the action of an anti-linear involution $*$, which we regard as CPT (see Ref. \cite{KTY1}). Using equivariance of the multiplication and (anti-)unitarity of $R(g)$ with respect to the inner product on the Hilbert space, it may be shown that $*$ commutes with $R(g)$ for all $g\in\cG$. With respect to the product on $A$, this map is an anti-automorphism. If $R(g)$ denotes the action of a time-reversing symmetry, $R(g)*$ is a unitary symmetry that reverses the orientation of space. Then
\begin{equation}
T(R(g)*a)=M(g)\overline{T(*a)}M(g)^{-1}=M(g)T(a)^TM(g)^{-1}.
\end{equation}
Moreover, since $*$ commutes with $R(g)$, the equivariance condition \eqref{equivariance} implies that $M(g)$ is unitary (up to a phase), so \eqref{proj1} and \eqref{proj} are equivalent (up to a coboundary). For the remainder of the paper, we suppress $*$ and simply write $R$ to denote a time-reversing or space-reversing symmetry.

\subsection{Invariants of fermionic SRE phases with anti-unitary symmetries}

As in the case of unitary symmetries, fermionic SRE systems at fixed points correspond to even algebras of the form $\End(U)$ and odd algebras of the form $\End(U_b)\otimes\CC\ell(1)$. However, when the symmetries may act anti-unitarily, the cohomology class characterizing the Morita class (and hence the SRE phase) is twisted.

We now discuss the meaning of the invariants $\alpha$, $\beta$, and $\gamma$ in the anti-unitary context, following the previous analysis. The form of the MPS conjugate wavefunction is \eqref{twistedMPS} as before. Consider the action of an anti-unitary symmetry $h\notin\cG_0$ on an MPS in the $g$-twisted ($g\in\cG_0$) sector:
\begin{align}\begin{split}
R(h)\cdot\Tr[Q(g)XT^i]\bra{i}&=\Tr[M(g)XM(h^{-1})(T^i)^TM(h^{-1})^{-1}]\bra{i}\\&=\Tr[M(h^{-1})^TX^TM(g)^TM(h^{-1})^{-1T}T^i]\bra{i}\\&\hspace{-1cm}=e^{2\pi i\omega(h,h^{-1})}\Tr[M(h^{-1})^TM(hg)^{-1}M(hg)X^T[M(h)M(g)^{-1T}]^{-1}T^i]\bra{i}\\&\hspace{-1cm}=e^{2\pi i(\omega(h,h^{-1})-\omega(h,g))}\Tr[[M(hg)M(h^{-1})^{-1T}]^{-1}[M(hg)X^TM(hg)^{-1}]T^i]\bra{i}\\&\hspace{-1cm}=e^{2\pi i(\omega(h,g^{-1})+\omega(hg^{-1},h^{-1})+\omega(g,g^{-1})-\omega(h,h^{-1}))}\Tr[Q(hg^{-1}h^{-1})[M(hg)X^TM(hg)^{-1}]T^i]\bra{i}\end{split}
\end{align}
where in the last line we use the fact that
\begin{equation}
\omega(hg^{-1}h^{-1},hgh^{-1})=-\omega(h,g)-\omega(hg,h^{-1})-\omega(hg^{-1},h^{-1})-\omega(h,g^{-1})-\omega(g^{-1},g)-2\omega(h^{-1},h),
\end{equation}
which can be verified by repeated application of the twisted cocycle condition. We see that under the action of an anti-unitary symmetry $h$,
\begin{enumerate}
\item The $g$-twisted sector maps to the $hg^{-1}h^{-1}$-twisted sector.
\item The operator $X$ is transposed, then conjugated by $M(hg)$.\footnote{If $X$ is Hermitian, this is the same as $X$ being conjugated by the anti-linear operator $Q(hg)$.}
\item States also pick up a phase of
\begin{equation}\label{antiunitaryphase}
e^{2\pi i(\omega(h,g^{-1})+\omega(hg^{-1},h^{-1})+\omega(g,g^{-1})-\omega(h,h^{-1}))}.
\end{equation}
\end{enumerate}
The phase matches Equation 4.12 of Ref. \cite{KT}. In particular, when $g$ acts on the R sector, it is
\begin{equation}
\sfrac{1}{2}\beta(g):=\omega(g,p)-\omega(p,g)+\omega(p,p),\qquad g\notin\cG_0.
\end{equation}
This phase satisfies $\beta(pg)=\beta(g)$, takes values in $\ZZ/2$, and, together with \eqref{beta}, is a $G_b$-cocycle. Refer to the appendix for a proof. When $\gamma=0$, this is the $g$-charge of the R sector. However, when $\gamma=1$, the charge receives an additional contribution from the transformation of $X$. Similar to in the unitary case detailed above, the total charge is the $\beta$-independent quantity $(-1)^{x(g)}$, so this interpretation of $\beta$ fails.

The invariant $\beta$ also has an interpretation in terms of edge states, like \eqref{Qgparity}.\footnote{If $g$ is anti-linear, the expression \eqref{Qgparity} is not invariant under the change of gauge $\omega\mapsto\omega+\delta\Lambda$.} A time-reversing symmetry $g\notin G_0$ maps $V$ to its dual space $V^*$, on which $p$ acts as $P^{-1}$, so the parity of $Q(g)$ is read off of
\begin{equation}
P^{-1}Q(g)P^{-1}=e^{i\pi\beta(\bar g)}Q(g),\qquad g\notin\cG_0
\end{equation}
A similar interpretation holds if $g$ reverses the orientation of space. Let $V^*\otimes V$ represent the tensor product of left and right edge state spaces. On this space, $g$ acts as
\begin{equation}
\psi_L\otimes\psi_R\mapsto Q(g)^{-1}(\psi_L\otimes\psi_R)^TQ(g)=Q(g)^{-1}\psi_R\otimes Q(g)\psi_L.
\end{equation}
Then $\beta$ appears as the result of acting by $P\otimes P^{-1}$, $g$, then $P^{-1}\otimes P$:
\begin{equation}
\psi\otimes 1\mapsto 1\otimes PQ(g)P\psi=e^{i\pi\beta(\bar g)}(1\otimes\psi).
\end{equation}

The meaning of $\alpha$ \eqref{alphadef} is more difficult to describe in Hamiltonian language.\footnote{In the Lagrangian picture, we expect $\alpha$ to be related to trivalent junctions of possibly orientation-reversing domain walls.} The lack of twisted sectors for anti-unitary symmetries means that $\alpha(g_b,h_b)$ has an interpretation as the phase due to fusing domain walls only when $g_b$ and $h_b$ are unitary. The rest of $\alpha$ appears in other places. It is convenient to first describe the invariant $\omega$. For two unitary symmetries $g,h\in\cG_0$, the phase $\omega(g,h)$ is due to fusing domain walls. It was shown in Ref. \cite{KT} that two extra families of phases -- which we now describe -- together with $\omega$ restricted to $\cG_0$, determine the full $\omega$ on $\cG$. The first family is the phases \eqref{antiunitaryphase} due to acting on the $g$-twisted sector by an anti-unitary symmetry $h$. The second family consists of the relative phases due to comparing, for each anti-unitary symmetry $g\notin\cG_0$, the crosscap state (see Refs. \cite{ShiozakiRyu,KT}) $\Tr[Q(g)Q(g)T^i]\bra{i}$ to the MPS state in the $g^2$-twisted sector. These phases have the simple form $\omega(g,g)$. Note that these data are not gauge invariant, and the equivalence classes of them under shifting $\omega$ by a twisted coboundary do not take a simple form. Now that we have described the full $\omega$, the full $\alpha$ can be recovered by restricting to $G_b$. As we demonstrate in the appendix, the result is a $G_b$ cochain whose $x$-twisted coboundary is $\beta\cup\rho$.

Finally, $\gamma$ is the fermion parity of the untwisted Ramond sector, as in the unitary case.

\section{The fermionic stacking law}\label{stackingsection}

Gapped fermionic phases form a commutative monoid under the operation of stacking. The result of stacking fixed point systems corresponding to algebras $A_1$ and $A_2$ is the system corresponding to the supertensor product $A_1\fotimes A_2$, defined by the multiplication law $(a_1\fotimes a_2)(b_1\fotimes b_2)=(-1)^{|a_2||b_1|}a_1b_1\fotimes a_2b_2$ \cite{FMPS,KTY2}. SRE phases are precisely those that are invertible under stacking, and so they form a group. The goal of this section is to derive this group structure on the set of SRE phases in terms of the invariants $\alpha$, $\beta$, and $\gamma$. Our plan is to follow the argument presented in Appendix D of Ref. \cite{KTY2}, while taking into account that $Q(g)$ is anti-linear when $x(g) = 1$. We summarize the results at the end of the section.

The following discussion relies on a result proven in the appendix: that one can choose a gauge such that the twisted cocycle $\omega$ is related to $\alpha$ and $\beta$ by, for all $g,h\in\cG$, where $\bar g$ is short for $b(g)$,
\begin{equation}
\omega(g,h) = \alpha(\bar g,\bar h) + \sfrac{1}{2}\beta(\bar g)t(h).
\end{equation}

There are three cases to consider: the stacking of 1) two even algebras, 2) an even and an odd algebra, 3) two odd algebras. When $\cG$ does not split, there are no odd algebras so we need only consider the first case.

\noindent\emph{Even-Even Stacking.}

Consider the even algebras $\End (U_1)$ and $\End (U_2)$. Their tensor product is $\End (U_1 \fotimes U_2)$, where $U_1 \fotimes U_2$ carries a projective representation $Q = Q_1 \fotimes Q_2$. Then
\begin{align}\begin{split}
Q(g) Q(h) &= \left(Q_1(g) \fotimes Q_2(g) \right) \left(Q_1(h) \fotimes Q_2(h) \right) \\
&= (-1)^{\beta_2(\bar g)\beta_1(\bar h)} Q_1(g) Q_1(h) \fotimes Q_2(g) Q_2(h) \\  
&= (-1)^{(\beta_2 \cup \beta_1) (\bar g,\bar h)} e^{2 \pi i (\alpha_1(\bar g,\bar h) +\sfrac{1}{2} \beta_1(\bar g) t(h))} e^{2 \pi i (\alpha_2(\bar g,\bar h) + \sfrac{1}{2} \beta_2(\bar g) t(h))} Q_1(gh) \fotimes Q_2(gh) \\ 
&= \exp \left( 2 \pi i (\alpha_1 + \alpha_2 + \sfrac{1}{2} \beta_2 \cup \beta_1 ) (\bar g,\bar h) +\sfrac{1}{2} (\beta_1 + \beta_2) (\bar{g}) t(h) \right) Q(gh).
\end{split}
\end{align}
Thus the invariants of the stacked phase are $\alpha = \alpha_1 + \alpha_2 + \sfrac{1}{2}(\beta_1  \cup \beta_2)$,\footnote{We have used the fact that $\beta_2 \cup \beta_1$ is cohomologous to $ \beta_1 \cup \beta_2$ in $\ZZ/2$.} and $\beta = \beta_1 + \beta_2$. Since the stacked algebra is again even, $\gamma= 0$. The presence of anti-unitary symmetries does not affect even-even stacking.

\noindent\emph{Even-Odd Stacking.}

Now consider the even algebra $A_1 = \End(U_1)$, where $U_1$ carries a projective representation $Q_1$ of $\cG$, and the odd algebra $A_2=\End(U_2) \otimes\CC\ell(1)$, where $U_2$ carries a projective representation $Q_2$ of $G_b$. Their tensor product $\End (U_1) \fotimes (\End(U_2) \otimes\CC\ell(1))$ is isomorphic as an algebra to the odd algebra $\End (U_1\otimes U_2) \otimes\CC\ell(1)$ by the map
\begin{equation}JW:M_1 \fotimes (M_2 \otimes \Gamma^m) \mapsto M_1 P^m \otimes M_2\otimes \Gamma^{m +|M_1|},\end{equation}
which has inverse
\begin{equation} JW^{-1} : M_1 \otimes M_2 \otimes \Gamma^m \mapsto M_1 P^{m+ |M_1|} \fotimes (M_2 \otimes \Gamma^{m + |M_1|}),
\end{equation}where the parity of $M_1$ is defined by $Q_1$: $P_1M_1P_1=(-1)^{|M_1|}M_1$. This isomorphism respects the $\ZZ/2$-grading defined by the standard action of fermion parity on even and odd algebras.

It remains to determine the $G_b$ action on the odd algebra. For $g\in\cG$ with $t(g)=0$,
\begin{align}\label{Gbaction}
\begin{split}
JW \circ g \circ JW^{-1} \cdot \left(M_1 \otimes M_2 \otimes \Gamma^m\right) &= JW \circ g \cdot \left(M_1 P^{m +|M_1|} \fotimes (M_2 \otimes \Gamma^{m +|M_1|})\right) \\ 
&\hspace{-3cm}= JW\cdot\left(Q_1(g) M_1 P^{m + |M_1|} Q_1(g)^{-1} \fotimes (Q_2(\bar{g}) M_2 Q_2(\bar{g})^{-1} \otimes (-1)^{(m + |M_1|) \beta_2(\bar{g})} \Gamma^{m + |M_1|})\right) \\
&\hspace{-3cm}= (-1)^{(m +|M_1|) (\beta_1(\bar g) + \beta_2(\bar g))} Q_1(g) M_1 Q_1(g)^{-1} \otimes Q_2(\bar{g}) M_2 Q_2(\bar{g})^{-1} \otimes \Gamma^{m}
\end{split}
\end{align}
In order to read off the invariants from this group action, we must rewrite it in the standard form by defining $\tilde{Q}_1(g) = Q_1(g) P^{\beta_1(g) + \beta_2(g)}$\footnote{Adding a phase factor to $\tilde{Q}$ would have shifted the resulting $2$-cocycle $\alpha$ by an irrelevant coboundary. For example, if we had chosen a factor $i^{\beta_1(g)}$ as in Ref. \cite{KTY2}, we would have gotten $\beta_1 \cup x$ instead of $\beta_1 \cup \beta_1$ in the final answer.} and $Q(g)=\tilde Q_1(g)\otimes Q_2(\bar g)$. Then, continuing from \eqref{Gbaction}, 
\begin{equation}
g \cdot (M_1 \otimes M_2 \otimes \Gamma^m) = (-1)^{m (\beta_1 (\bar g) + \beta_2(\bar g))} (\tilde{Q}_1(g) \otimes Q_2(\bar g)) M_1 \otimes M_2 (\tilde{Q}_1(g)^{-1} \otimes Q_2(\bar g)^{-1} ) \otimes \Gamma^m,
\end{equation}
from which we read off the stacked invariant $\beta = \beta_1 + \beta_2$. And
\begin{align}
\begin{split}
Q(g)Q(h) &= \left( \tilde{Q}_1(g) \otimes Q_2(\bar g) \right) \left( \tilde{Q}_1(h) \otimes Q_2(\bar h) \right) \\ 
&= Q_1(g) P^{\beta_1(\bar g) + \beta_2(\bar g)} Q_1(h) P^{\beta_1(\bar h) + \beta_2(\bar h)}\otimes Q_2(\bar g) Q_2(\bar h) \\
&= (-1)^{\beta_1(\bar h)(\beta_1(\bar g) + \beta_2(\bar g))} Q_1(g) Q_1(h) P^{\beta_1(\bar{gh}) + \beta_2(\bar{gh})} \otimes Q_2(\bar g) Q_2(\bar h) \\
&=  \exp \left(2 \pi i( \alpha_1(g,h) + \alpha_2(g,h) +\sfrac{1}{2} (\beta_2 \cup \beta_1)(g,h) + \sfrac{1}{2} (\beta_1 \cup \beta_1) (g,h) ) \right) Q(gh),
\end{split}
\end{align}
from which we see $\alpha = \alpha_1 + \alpha_2 + \sfrac{1}{2} \beta_1 \cup \beta_2 + \sfrac{1}{2} \beta_1 \cup \beta_1$. There is no asymmetry: the $\sfrac{1}{2} \beta_1 \cup \beta_1$ term always comes from the $\beta$ of the even algebra.\footnote{Note that while $\beta_1 \cup \beta_1$ is an ordinary coboundary and hence could be ignored for phases without time-reversal, it is not a twisted coboundary and so cannot be ignored when time-reversing symmetries are present. By adding a twisted coboundary, we can put it in the form $\beta_1 \cup x$, which makes the dependence on time-reversal symmetry manifest.} Finally, $\gamma = 1$ since the stacked algebra is odd.

\noindent\emph{Odd-Odd Stacking.}

Consider the odd algebras $A_1 = \End(U_1) \otimes\CC\ell(1)$, where $U_1$ carries a projective representation $Q_1$ of $G_b$, and $A_2 = \End (U_2) \otimes \CC\ell(1)$, where $U_2$ carries a projective representation $Q_2$ of $G_b$. Their tensor product is given by $A_1 \fotimes A_2 \simeq \End (U_1 \otimes U_2 \otimes \CC^{1|1})$, since $\CC\ell(1) \fotimes\CC\ell(1) \simeq\CC\ell(2) \simeq \End (\CC^{1|1})$, via an isomorphism  
\begin{equation}\label{oddoddisom}
(M_1 \otimes \Gamma_1^m) \fotimes (M_2 \otimes \Gamma_2^n) \mapsto M_1 \otimes M_2 \otimes \sigma_1^m \sigma_2^n,
\end{equation}
where $\sigma_1$ and $\sigma_2$ are any two distinct Pauli matrices.
With respect to the action of fermion parity on the $\End(\CC^{1|1})$ factor as conjugation by $\sigma_3=-i\sigma_1\sigma_2$, this map is an isomorphism of $\ZZ/2$-graded algebras.

One choice\footnote{Again, had we chosen a different $G_b$-action on $\CC^{1|1}$ compatible with the action $g_b:\Gamma_i \mapsto (-1)^{\beta_i(g_b)}\Gamma_i$ on the $\CC\ell(1)$ factors, the $2$-cocycle $\alpha$ would be shifted by a twisted coboundary.} of $G_b$-action $Q$ on $U_1 \otimes U_2 \otimes \CC^2$, with respect to which \eqref{oddoddisom} is equivariant, is
\begin{align}
g: u_1 \otimes u_2 \otimes v \mapsto Q_1(\bar g) u_1 \otimes Q_2(\bar g) u_2 \otimes \sigma_1^{\beta_2(\bar g)} \sigma_2^{\beta_1(\bar g)} K^{x(\bar g)} v,
\end{align}
for $g\in\cG$ with $t(g)=0$, where $K$ denotes complex conjugation in a basis in which $\sigma_1$ and $\sigma_2$ are real. Then
\begin{align}
\begin{split}
Q(\bar g) Q(\bar h) &= \left( Q_1(\bar g) \otimes Q_2(\bar g) \otimes \sigma_1^{\beta_2(\bar g)} \sigma_2^{\beta_1(\bar g)} K^{x(\bar g)} \right) \left( Q_1(\bar h) \otimes Q_2(\bar h) \otimes \sigma_1^{\beta_2(\bar h)} \sigma_2^{\beta_1(\bar h)} K^{x(\bar h)} \right)  \\
&= e^{ 2 \pi i \alpha_1(\bar g,\bar h) } Q_1(\bar{gh}) \otimes e^{ 2 \pi i \alpha_2 (\bar g,\bar h) } Q_2(\bar{gh}) \otimes (-1)^{\beta_1(\bar g) \beta_2(\bar h) } \sigma_1^{\beta_1(\bar{gh})} \sigma_2^{\beta_2(\bar{gh})} K^{x(\bar{gh})}  \\
&= \exp \left( 2 \pi i (\alpha_1 + \alpha_2 + \sfrac{1}{2} \beta_1 \cup \beta_2)(\bar g,\bar h) \right) Q(\bar{gh}),
\end{split}
\end{align}
from which we see that $\alpha = \alpha_1 + \alpha_2 + \sfrac{1}{2} \beta_1 \cup \beta_2$. Since $U_1\otimes U_2$ is purely even, the parity of $Q$ comes from
\begin{align}
\begin{split}
P \sigma_1^{\beta_2(\bar g)} \sigma_2^{\beta_1(\bar g)} K^{x(\bar g)} P &= (-i \sigma_1 \sigma_2) \sigma_1^{\beta_2(\bar g)} \sigma_2^{\beta_1(\bar g)} K^{x(\bar g)} (-i \sigma_1 \sigma_2) \\ 
&= (-1)^{\beta_2(\bar g) + \beta_1(\bar g)} \sigma_1^{\beta_2(\bar g)} \sigma_2^{\beta_1(\bar g)} K^{x(\bar g)} (-1)^{x(\bar g)}.
\end{split}
\end{align}
We read off $\beta = \beta_1 + \beta_2 + x$. Finally, the stacked algebra is even, so $\gamma =0$.

\noindent\hspace{0.02\textwidth}\fbox{
\begin{minipage}{0.94\textwidth}
\vspace{1mm}
\noindent\emph{In Summary.}
\vspace{0.2cm}

The stacking law for the invariants $(\alpha,\beta,\gamma)$ is given by
\begin{align}\label{stackinglaw}
\begin{split}
&(\alpha_1, \beta_1, 0) \cdot (\alpha_2, \beta_2, 0) = (\alpha_1 + \alpha_2 + \sfrac{1}{2} \beta_1 \cup \beta_2, \beta_1 + \beta_2, 0) \\
&(\alpha_1, \beta_1, 0) \cdot (\alpha_2, \beta_2,1) = (\alpha_1 +\alpha_2 + \sfrac{1}{2} \beta_1 \cup \beta_2 +\sfrac{1}{2} \beta_1 \cup \beta_1, \beta_1 + \beta_2, 1) \\
&(\alpha_1, \beta_1, 1) \cdot (\alpha_2, \beta_2, 1) = (\alpha_1 + \alpha_2 + \sfrac{1}{2} \beta_1 \cup \beta_2,\beta_1 + \beta_2 + x, 0).
\end{split}
\end{align}
\vspace{-0.1cm}
\end{minipage}
}

While this stacking law looks different from that of Ref. \cite{FMPS}, the two results are actually consistent. The precise relationship between various formulations of the stacking law is discussed in Ref. \cite{duality23}.

This group law inherits the properties of commutativity and associativity from the tensor product of algebras. When $\cG$ does not split, $\gamma$ is not present, and the stacking law is simply 
\begin{equation}
(\alpha_1, \beta_1) \cdot (\alpha_2, \beta_2) = (\alpha_1 + \alpha_2 + \sfrac{1}{2} \beta_1 \cup \beta_2 , \beta_1 + \beta_2).
\end{equation}

We emphasize that while data $[\alpha,\beta]$ are equivalent to $[\omega]\in H^2(\cG,U(1)_T)$, the group structure on $H^2(\cG,U(1)_T)$ differs from \eqref{stackinglaw}. On the other hand, the stacking of bosonic SRE phases, which are also characterized by classes $[\omega]$, is described by the usual group structure on $H^2(\cG,U(1)_T)$.

\section{Examples}\label{examplesection}

\subsection{Class BDI fermions: $\cG=\ZZ_2^F\times\ZZ_2^T$}

Let us consider SRE phases with symmetry $\cG=\ZZ_2^F\times\ZZ_2^T$. The two classes $\alpha\in H^2(\ZZ_2^T,U(1)_T)=\ZZ/2$, two classes $\beta\in H^1(\ZZ_2^T,\ZZ/2)=\ZZ/2$, and two classes $\gamma\in\ZZ/2$ make for a total of eight phases. A straightforward application of the general stacking law \eqref{stackinglaw} reveals that these phases stack like the cyclic group $\ZZ/8$. In this section, we will reproduce this group law by exploiting the relationship between $\cG$-equivariant algebras, real super-division algebras, and Clifford algebras, which have Bott periodicity $\ZZ/8$.

We begin by describing simple $\cG$-equivariant algebras. The matrix algebra $M_2\CC$ represents the sole Morita class of simple complex algebras. This algebra has a unitary structure $*$ given by conjugate transposition. Its action fixes a basis $\{\mathds{1},X,Y,Z=-iXY\}$. On $C\ell_2\CC\simeq M_2\CC$, $*$ acts by Clifford transposition and complex conjugation of coefficients with respect to a pair of generators that square to $+1$.

There are two distinct real structures on $M_2\CC$ given by complex conjugation $T$ on the second component of $M_2\CC\simeq M_2\RR\otimes_\RR\CC$ and $M_2\CC\simeq\HH\otimes_\RR\CC$. The unitary structure $*$ of $M_2\CC$ acts by transposition on $M_2\RR$, complex conjugation on $\CC$, and inversion of the generators $\hat\imath$ and $\hat\jmath$ of $\HH$; that is, its fixed bases are
\begin{equation}\{\mathds{1}\otimes 1,X\otimes 1,iY\otimes i,Z\otimes 1\}\in M_2\RR\otimes_\RR\CC,\qquad\text{and}\qquad\{\mathds{1}\otimes 1,\hat\imath\otimes i,\hat\jmath\otimes i,\hat k\otimes i\}\in\HH\otimes_\RR\CC.\end{equation}
These bases have the same $T$-eigenvalues as they do $*T$-eigenvalues, where $*T$ acts as transposition on $M_2\RR$ and inversion of generators on $\HH$. Under the algebra isomorphisms $M_2\RR\simeq C\ell_{1,1}\RR\simeq C\ell_{2,0}\RR$ and $\HH\simeq C\ell_{0,2}\RR$, $*T$ acts by inverting the generators and products of generators that square to $-1$.

Let us derive the invariants $\alpha$ of these real structures. Pulled back from $M_2\RR\otimes_\RR\CC$ to $M_2\CC$, $T$ acts like complex conjugation and $*T$ like transposition; that is $M(t)=\mathds{1}$. Then\begin{equation}M(t)M(t)^{-1T}=\mathds{1},\end{equation}which means $\alpha(t,t)=\omega(t,t)=0$. Pulled back from $\HH\otimes_\RR\CC$, $T$ acts like complex conjugation and conjugation by $Y$, while $*T$ acts like transposition and conjugation by $Y$; that is $M(t)=Y$. Then $\alpha(t,t)=\sfrac{1}{2}$ since\begin{equation}M(t)M(t)^{-1T}=e^{i\pi}\mathds{1}.\end{equation}

By the Skolem-Noether theorem, a superalgebra structure on $M_2\CC$ is given by conjugation by an element that squares to one. If this element is $\mathds{1}$, the $\ZZ/2$-grading is purely even; otherwise, it has two even dimensions and two odd. All structures of the latter type are isomorphic in the absence of a real structure.

In the presence of the real structure $M_2\RR\otimes_\RR\CC$, there are three distinct gradings. First, there is the purely even grading, given by $P=\mathds{1}$. This structure has $\sfrac{1}{2}\beta(t)=\omega(t,p)=0$. Second, there is conjugation by $Z$ (or $X$), which gives $M_2\RR$ the superalgebra structure of $C\ell_{1,1}\RR$. Again, $\beta(t)=0$ since $P=Z$ means\begin{equation}PM(t)P^T=Z\mathds{1}Z^T=\mathds{1}.\end{equation}The matching of the invariants alludes to the fact that the real superalgebra structures $M_2\RR$ and $C\ell_{1,1}\RR$ are graded Morita equivalent. Third, there is conjugation by $Y$; that is, $P=Y$. Then $\beta(t)=1$ since\begin{equation}PM(t)P^T=Y\mathds{1}Y^T=e^{i\pi}\mathds{1}.\end{equation}The corresponding real Clifford algebra is $C\ell_{2,0}\RR$ and represents a distinct Morita class.

On the real structure $\HH\otimes_\RR\CC$, there are two distinct gradings. First, there is the purely even grading $P=\mathds{1}$, which has $\beta(t)=0$. The second grading is given by conjugation by $Z$ (or $X$ or $Y$) on $M_2\CC$ and gives $\HH$ the superalgebra structure of $C\ell_{0,2}\RR$. Then $\beta(t)=1$ since\begin{equation}PM(t)P^T=ZYZ^T=e^{i\pi}Y.\end{equation}

Now consider algebras of the form $M_2\CC\otimes C\ell_1\CC$. The second component $C\ell_1\CC$ has a unitary structure $*$ given by complex conjugation of coefficients of the generator $\Gamma$ that squares to $+1$. There are two distinct real structures on $C\ell_1\CC$ given by complex conjugation $T$ on the second component of $C\ell_{1,0}\RR\otimes_\RR\CC$ and $C\ell_{0,1}\RR\otimes_\RR\CC$. The unitary structure $*$ of $C\ell_1\CC$ acts by complex conjugation on $\CC$ and inversion of generators that square to $-1$; that is, the fixed bases are $\{1\otimes 1,\gamma\otimes 1\}$ for $C\ell_{1,0}\RR\otimes_\RR\CC$ and $\{1\otimes 1,\gamma\otimes i\}$ for $C\ell_{0,1}\RR\otimes_\RR\CC$. The map $*T$ is trivial on $C\ell_{1,0}\RR$ and inversion of the generator on $C\ell_{0,1}\RR$. Therefore, pulled back from $C\ell_{1,0}\RR\otimes_\RR\CC$ to $C\ell_1\CC$, $T$ is complex conjugation and $*T$ is trivial. From $C\ell_{0,1}\RR\otimes_\RR\CC$, $*T$ is inversion of $\Gamma$.

As discussed in Section \ref{fermionicMPS} and in Ref. \cite{KTY2}, we need only to consider a single $\ZZ/2$-grading on $M_2\CC\otimes C\ell_1\CC$ - the one where $M_2\CC$ is purely even and the generator of $C\ell_1\CC$ is odd. The algebra $M_2\CC\otimes C\ell_1\CC$ has four real structures: a choice of $M_2\RR$ or $\HH$ for the first component and $C\ell_{1,0}\RR$ or $C\ell_{0,1}\RR$ for the second. As was true for even algebras, the first choice determines whether $M(t)$ is $\mathds{1}$ or $Y$; that is, whether $\alpha(t,t)$ is $0$ or $\sfrac{1}{2}$. The second choice determines whether $*T$ inverts the odd generator; this is $\beta(t)$.

Due to the Morita equivalence $M_2\RR\sim\RR$, several of the eight Morita classes are represented by algebras of lower dimension; for example, $C\ell_{1,0}\RR$ instead of $M_2\RR\otimes_\RR C\ell_{1,0}\RR$. Up to this substitution, the eight real-structured superalgebras we found are complexifications of the eight \emph{central real super-division algebras} -- real superalgebras with center $\RR$ that are invertible under supertensor product up to graded Morita equivalence \cite{wall,trimble}. They constitute a set of representatives of the eight graded Morita classes of real superalgebras. These algebras appear in second column of Figure \ref{fig:Z8table}, next to their invariants in the third column.

Another set of Morita class representatives is the Clifford algebras $C\ell_{n,0}\RR$. In terms of these algebras, stacking is simple, as
\begin{equation}
C\ell_{n,0}\RR\fotimes C\ell_{m,0}\RR\simeq C\ell_{n+m,0}\RR
\end{equation}
and
\begin{equation}
C\ell_{n,0}\RR\sim C\ell_{m,0}\RR\qquad\text{for }n\equiv m\text{ mod }8.
\end{equation}
Each central super-division algebra can be matched with the Clifford algebra $\CC\ell_{n,0}\RR$, $n<8$ in its Morita class \cite{trimble}, as in the first column of Figure \ref{fig:Z8table}. This determines a $\ZZ/8$ stacking law on central super-division algebras and their invariants that agrees with the more general law \eqref{stackinglaw}.

Physically speaking, the $\ZZ/8$ classification is generated by the time-reversal-invariant Majorana chain \cite{fidkit,FK2}. While the symmetry protects pairs of dangling Majorana zero modes from being gapped out, turning on interactions can gap out these modes in groups of eight. Fidkowski and Kitaev formulate their stacking law in terms of three invariants that are equivalent to $\alpha$, $\beta$, and $\gamma$. Their results match ours.

For contrast, we list the invariants of the corresponding bosonic phases in the rightmost column of Figure \ref{fig:Z8table}. There, $H$ denotes the subgroup of unbroken symmetries and $\omega$ denotes $2$-cocycle characterizing the SPT order. These invariants can be obtained from the fermionic invariants \cite{KTY2}. We observe that invertibility is not preserved by bosonization; in particular, only the fermionic SREs with $\gamma=0$ become bosonic SREs. The four bosonic SRE phases have a $\ZZ/2\times\ZZ/2$ stacking law. We also include the two \emph{non-central} super-division algebras $\CC$ and $\CC\ell_1$ at the bottom of the table. These correspond to symmetry-breaking (SB) phases.

\begin{figure}[h]
\centering
\begin{tabular}{ c | l | c | l | l | l }
 $C\ell_{n,0}$ & $A_\text{div}$ & $\alpha,\beta,\gamma$ & fermionic & bosonic & $(H,\omega)$ \\
 \hline
 & & & & & \\
  $0$ & $\RR$ & $0,0,0$ & trivial & trivial & $(\cG,0)$ \\
    $1$ & $C\ell_{1,0}$ & $0,0,1$ & SRE & SB & $(\ZZ_2^T,0)$ \\
  $2$ & $C\ell_{2,0}$ & $0,1,0$ & SRE & SPT & $(\cG,\omega_1)$ \\
  $3$ & $\HH\otimes C\ell_{0,1}$ & $1,1,1$ & SRE & mixed & $(\ZZ_2^\text{diag},\alpha)$ \\
   $4$ & $\HH$ & $1,0,0$ & SRE & SPT & $(\cG,\omega_2)$ \\
     $5$ & $\HH\otimes C\ell_{1,0}$ & $1,0,1$ & SRE & mixed & $(\ZZ_2^T,\alpha)$ \\
  $6$ & $C\ell_{0,2}$ & $1,1,0$ & SRE & SPT & $(\cG,\omega_1+\omega_2)$ \\
   $7$ & $C\ell_{0,1}$ & $0,1,1$ & SRE & SB & $(\ZZ_2^\text{diag},0)$ \\
  & & & & & \\
   - & $\CC$ & - & SB & SB & $(\ZZ_2^F,0)$ \\
   - & $\CC\ell_1$ & - & SB & SB & $(1,0)$ \\
  & & & & & \\
  \hline
\end{tabular}
\caption{the $10$-fold way of $\ZZ_2^F\times\ZZ_2^T$-symmetric fermionic phases}
\label{fig:Z8table}
\end{figure}

\subsection{Class DIII fermions: $\cG=\ZZ_4^{FT}$}

In the following, $\cG = \ZZ_4^{FT}$ denotes the non-trivial extension of $G_b = \ZZ_2^T$ by fermion parity. Let us consider fermionic SRE phases with this symmetry. There are two distinct classes $\beta \in H^1(\ZZ_2^T,\ZZ/2)$, determined by $\beta(t) = 0$ and $\beta(t) = 1$. The trivial $\beta$ has a single $\alpha$, the trivial one, that satisfies $\delta_T \alpha = \sfrac{1}{2} \beta \cup \rho$, up to the proper equivalence.\footnote{The cocycle $\alpha(t,t) = \sfrac{1}{2}$ is nontrivial in $H^2(G_b, U(1)_T)$ but is trivialized by adding a $2$-coboundary on $\cG$ satisfying the proper conditions. See the Appendix for details.} The nontrivial $\beta$ also has a single compatible $\alpha$, up to equivalence: $\alpha(t,t) = \sfrac{1}{4}$.

The trivial phase is represented by the algebra $\CC$ with trivial actions of $p$ and $t$, as always. For the nontrivial phase, consider $A=\End(U)$, where $P$ and $T$ act on $U$ as\begin{equation}P=\left(\begin{array}{cc}1&0\\0&-1\end{array}\right)\qquad\text{and}\qquad M(t)=\left(\begin{array}{cc}0&1\\-i&0\end{array}\right).\end{equation}Then the invariants can be recovered:
\begin{align}M(t)M(t)^{-1T}=e^{2\pi i/4}P\qquad&\Rightarrow\qquad\alpha(t,t)=\sfrac{1}{4}\\PM(t)P^T=e^{i\pi}M(t)\qquad&\Rightarrow\qquad\beta(t)=1\end{align}

According to the rule \eqref{stackinglaw}, stacking two copies of this phase results in the trivial phase:
\begin{equation}(\sfrac{1}{4},1)\cdot (\sfrac{1}{4},1)=(\sfrac{1}{4}+\sfrac{1}{4}+\sfrac{1}{2}\cdot 1\cdot 1,1+1)=(0,0).\end{equation}We find that fermionic SRE phases with symmetry $\ZZ_4^{FT}$ have a $\ZZ/2$ classification, in agreement with the condensed matter literature \cite{tenfold,periodic}. The nontrivial phase appears as a Majorana chain with two dangling modes protected by the symmetry.

\subsection{Unitary $\ZZ/2$ symmetry}

As a last set of examples, let us consider systems with a unitary bosonic symmetry group $G_0=\ZZ/2$, in addition to time-reversal and fermion parity. There are many ways to organize these symmetries into a full symmetry class $(\cG,p,x)$. Here, we consider the five abelian possibilities, which are listed with their fermionic and bosonic phase classifications in Figure \ref{fig:moreexamples}. The first three have $G_b=\ZZ_2\times\ZZ_2^T$, the last two $G_b=\ZZ_4^T$. In the two cases where the central extension of $G_b$ by $\ZZ_2^F$ splits, we use a superscript $\gamma$ to denote the subgroup of the fermionic classification that contains the odd phases.\footnote{An earlier version of this paper incorrectly stated the group of fermionic phases with $\ZZ_2^T\times\ZZ_4^F$ symmetry as $\ZZ_4$.}

\begin{figure}[h]
\centering
\begin{tabular}{ c || c | c }
 symmetry class & fermionic & bosonic \\
 \hline
  & & \\
   $\ZZ_2\times\ZZ_2^T\times\ZZ_2^F$ & $\ZZ_4\times\ZZ_8^\gamma$ & $(\ZZ_2)^4$ \\
    $\ZZ_2\times\ZZ_4^{FT}$ & $\ZZ_2\times\ZZ_2$ & $\ZZ_2\times\ZZ_2$ \\
      $\ZZ_2^T\times\ZZ_4^F$ & $\ZZ_2\times\ZZ_2$ & $\ZZ_2\times\ZZ_2$ \\
    $\ZZ_2^F\times\ZZ_4^T$ & $\ZZ_2\times\ZZ_4^\gamma$ & $\ZZ_2\times\ZZ_2$ \\
    $\ZZ_8^{FT}$ & $\ZZ_2$ & $\ZZ_2$ \\
   & & \\
  \hline
\end{tabular}
\caption{fermionic phases with unitary and anti-unitary symmetries}
\label{fig:moreexamples}
\end{figure}


\appendix

\section{Relations between bosonic and fermionic invariants}

\begin{lemma}\label{betalemma}
For a twisted cocycle $\omega\in Z^2(\cG,U(1)_T)$, the $1$-cochain defined by
\begin{equation}\label{fullbeta}
\sfrac{1}{2}\beta(g):=\omega(g,p)-\omega(p,g)+x(g)\omega(p,p)=\left\{\begin{array}{lr}\omega(g,p)-\omega(p,g)&g\in G_0\\\omega(g,p)-\omega(p,g)+\omega(p,p)&g\notin G_0\end{array}\right.
\end{equation}
is gauge-invariant, satisfies $\beta(gp)=\beta(g)$, takes values in $\{0,\sfrac{1}{2}\}$, and defines a $G_b$-cocycle.
\end{lemma}

\begin{proof}
First, $\sfrac{1}{2}\beta(g)$ picks up a factor of
\begin{align}
\begin{split}
(L(g) + (-1)^{x(g)} L(p) - L(gp)) -(L(p) + &L(g) - L(gp)) + x(g) (L(p) + L(p) - L(1))\\
&= - 2 x(g) L(p) + 2x(g) L(p) - x(g) L(1)= 0
\end{split}
\end{align}
under a transformation $\omega \mapsto \omega + \delta_T L$ for some $1$-cochain $L$ of $\cG$ satisfying $L(1)=0$.\footnote{This condition on $L$ ensures that $Q(1)=\mathds{1}$ is preserved.}

Second,
\begin{align}
\begin{split}
\sfrac{1}{2} \beta(gp) &= \omega(gp,p) - \omega(p,gp) + x(gp) \omega(p,p) \\&= \omega(g,p) - \omega(p,g) + x(g) \omega(p,p) -\delta_T \omega(p,g,p) \\&= \sfrac{1}{2}  \beta(g).
\end{split}
\end{align}

Third, 
\begin{align}
\begin{split}
\omega(g,p)-\omega(p,g)&=(-1)^{x(g)}\omega(p,p)-\omega(g,p)-\omega(p,gp)-(\delta_T\omega)(g,p,p)+(\delta_T\omega)(p,g,p)\\&=(-1)^{x(g)}\omega(p,p)-\omega(g,p)-\omega(p,p)+\omega(p,g)+(\delta_T\omega)(p,p,g)\\&=-\omega(g,p)+\omega(p,g)-(1-(-1)^{x(g)})\omega(p,p)
\end{split}
\end{align}
means that $\sfrac{1}{2} \beta$ takes values in the $\ZZ/2$ subgroup of $U(1)$.

Therefore $\sfrac{1}{2}\beta$ defines a $\beta\in C^1(G_b,\ZZ/2)$. Let $g_b,h_b\in G_b$ and choose any lifts $g,h$ to $\cG$. Fourth,       
\begin{align}
\begin{split}
(\delta \beta)(g_b,h_b) &= \sfrac{1}{2} (\beta(g) + \beta(h) - \beta(gh p^{\rho(\bar{g},\bar{h})})) \\&= \sfrac{1}{2} (\beta(g) + \beta(h) - \beta(gh)) \\ 
&= \omega(g,p) - \omega(p,g) + x(g) \omega(p,p) + \omega(h,p) - \omega(p,h) + x(h) \omega(p,p) \\&\qquad\qquad - \omega(gh,p) + \omega(p,gh) -x(gh) \omega(p,p) \\
&= \omega(g,p) - \omega(p,g) + \omega(h,p) - \omega(p,h) +\omega(g,h) - \omega(g,hp) - (-1)^{x(g)} \omega(h,p)  \\
&\qquad\qquad - \omega(g,h) +\omega(p,g) + \omega(pg,h) + 2 x(g) x(h) \omega(p,p)  \\
&= \omega(g,p) + 2x(g) \omega(h,p) - \omega(p,h) + (-1)^{x(g)} \omega(p,h) - \omega(g,p)  + 2 x(g) x(h) \omega(p,p)  \\
&= 2x(g) (\omega(h,p) - \omega(p,h) + x(h) \omega(p,p))  \\
&= 2x(g) \cdot \sfrac{1}{2} \beta(h)\\&= 0. 
\end{split}
\end{align}

\end{proof}

\begin{lemma}\label{gaugelemma}

Each cohomology class $H^2(\cG,U(1)_T)$ contains an element $\omega$ that satisfies, for all $g,h\in\cG$,
\begin{align}
\omega(pg, h) &= \omega(g,h) \label{gauge1}\\ \omega(g,ph) &= \omega(g,h) + \omega(g,p).\label{gauge2}
\end{align}

\end{lemma}

\begin{proof}
For an arbitrary $2$-cocycle $W \in Z^2(\cG, U(1)_T)$, define
\begin{equation}
\omega = W - \delta_T L
\end{equation}
where $L \in C^1(\cG, U(1)_T)$ satisfies
\begin{align}\label{gaugeL}
\begin{split}
L(1) &= 0 \\  L(p) &=  \sfrac{1}{2} W(p,p) \textrm{\ \ or \ }  \sfrac{1}{2} W(p,p) + \sfrac{1}{2} \\L(p\bar{g}) &= L(\bar{g}) - W(p,\bar{g}) + L(p).
\end{split}
\end{align}
Here, we abuse notation by letting $\bar g$ denote a $g\in\cG$ with $t(g)=0$. This implies $L(p) = \sfrac{1}{2} W(p,p)$. We have fixed $L(p\bar{g})$ in terms of $L(\bar{g})$ and $L(p)$ but left $L(\bar{g})$ undetermined, while $L(p)$ is fixed up to a $\sfrac{1}{2}$.  

We see that
\begin{equation}\label{cond1}
\omega(p,p) = W(p,p) - (-1)^{x(p)} L(p) - L(p) + L(1) = W(p,p) - 2 \cdot \sfrac{1}{2} W(p,p) = 0
\end{equation}
and
\begin{equation}\label{cond2}
\omega(p,\bar{g}) = W(p, \bar{g}) - \left( (-1)^{x(p)} L(\bar{g}) + L(p) - L(p\bar{g}) \right) =W(p, \bar{g}) - W(p,\bar{g}) +\sfrac{1}{2} W(p,p) -\sfrac{1}{2} W(p,p) =  0.
\end{equation}

Next we show that any $\omega$ satisfying \eqref{cond1} and \eqref{cond2} must also satisfy the gauge conditions \eqref{gauge1} and \eqref{gauge2}. First,
\begin{equation}
\omega(p, p\bar{g}) = -\delta_T \omega (p, p, \bar{g}) - (-1)^x(p) \omega(p, \bar{g}) + \omega(p,p) + \omega(1, \bar{g}) = 0.
\end{equation}
Similarly, computing $0 = \delta_T \omega (p, \bar{g}, \bar{h})$ shows that $\omega(\bar{g}p, \bar{h}) = \omega(\bar{g}, \bar{h})$ and computing $0 = \delta_T \omega (p, \bar{g}, \bar{h}p)$ shows that $\omega (\bar{g}p, \bar{h}p) = \omega(\bar{g}, \bar{h}p)$.  Putting these together, we see that \eqref{gauge1} is satisfied.

Now we compute $0 =\delta_T \omega (\bar{g}, p, \bar{h})$ which shows that $\omega(\bar{g},p \bar{h}) = \omega(\bar{g}, \bar{h}) + \omega (\bar{g}, p)$ and $0 = \delta_T \omega(\bar{g}, p, p\bar{h})$ which shows that $\omega(\bar{g}, \bar{h}) = \omega(\bar{g}, p \bar{h}) + \omega(\bar{g},p).$ Putting these together, we see that \eqref{gauge2} is satisfied.

\end{proof}

\begin{lemma}\label{isomlemma}
Given a trivialization $t$, the map
\begin{equation}
\omega(g,h)=\alpha(\bar g,\bar h)+\sfrac{1}{2}\beta(\bar g)t(h)
\end{equation}
defines a bijection from pairs $(\alpha,\beta)\in C^2(G_b,U(1)_T)\times C^1(G_b,\ZZ/2)$ that satisfy $\delta_T\alpha=\sfrac{1}{2}\beta\cup\rho$ and $\delta\beta=0$ (where $\sfrac{1}{2}\beta$ is regarded as a $U(1)_T$-valued cocycle) to twisted cocycles $\omega\in Z^2(\cG,U(1)_T)$ that satisfy \eqref{gauge1} and \eqref{gauge2} for all $g,h\in\cG$. In particular, for all $g_b,h_b\in G_b$, this map has an inverse
\begin{align}\label{inverse}
\begin{split}\alpha(g_b,h_b)&=\omega(s(g_b),s(h_b))\\\sfrac{1}{2}\beta(g_b)&=\omega(s(g_b),p).
\end{split}
\end{align}
\end{lemma}

\begin{proof}
First we show that $\omega$ is a twisted cocycle:
\begin{align}
\begin{split}
(\delta_T\omega)(g,h,k)&=(-1)^{x(g)}\omega(h,k)+\omega(g,hk)-\omega(g,h)-\omega(gh,k)\\&=(-1)^{x(g)}\alpha(\bar h,\bar k)+\alpha(\bar g,\bar{hk})-\alpha(\bar g,\bar h)-\alpha(\bar{gh},\bar k)\\&\qquad\qquad+\sfrac{1}{2}(-1)^{x(g)}\beta(\bar h)t(k)+\sfrac{1}{2}\beta(\bar g)t(hk)-\sfrac{1}{2}\beta(\bar g)t(h)-\sfrac{1}{2}\beta(\bar{gh})t(k)\\&=(\delta_T\alpha)(\bar g,\bar h,\bar k)+\sfrac{1}{2}(\delta_T\beta)(\bar g,\bar h)t(k)-\sfrac{1}{2}\beta(\bar g)(\delta t)(h,k)\\&=0.
\end{split}
\end{align}
Next we verify that $\omega$ satisfies the gauge conditions: 
\begin{equation}
\omega(pg,h)=\alpha(\bar{pg},\bar h)+\sfrac{1}{2}\beta(\bar{pg})t(h)=\alpha(\bar g,\bar h)+\sfrac{1}{2}\beta(\bar g)t(h)=\omega(g,h)
\end{equation}
\begin{equation}
\omega(g,ph)=\alpha(\bar g,\bar{ph})+\sfrac{1}{2}\beta(\bar g)t(ph)=\alpha(\bar g,\bar h)+\sfrac{1}{2}\beta(\bar g)(t(h)+1)=\omega(g,h)+\omega(g,p).
\end{equation}

Then we check the conditions for $\alpha$ and $\beta$. For these two calculations, let $\bar g$ denote $g_b$ and $g$ denote $s(g_b)$. Note that $s(\bar{gh})=p^{\rho(\bar g,\bar h)}gh$. Then
\begin{align}
\begin{split}
(\delta_T\alpha)(\bar g,\bar h,\bar k)&=(-1)^{x(g)}\alpha(\bar h,\bar k)+\alpha(\bar g,\bar{hk})-\alpha(\bar g,\bar h)-\alpha(\bar{gh},\bar k)\\&=(-1)^{x(g)}\omega(h,k)+\omega(g,p^{\rho(\bar h,\bar k)}hk)-\omega(g,h)-\omega(p^{\rho(\bar g,\bar h)}gh,k)\\&=(\delta_T\omega)(g,h,k)+\left\{\text{terms of the form }\omega(p^-,-)\right\}+\omega(g,p^{\rho(\bar h,\bar k)})\\&\qquad\qquad+(\delta_T\omega)(p^{\rho(\bar g,\bar h)},gh,k)-(\delta_T\omega)(g,p^{\rho(\bar h,\bar k)},hk)+(\delta_T\omega)(p^{\rho(\bar h,\bar k)},g,hk)\\&=\sfrac{1}{2}\beta(\bar g)\rho(\bar h,\bar k).
\end{split}
\end{align}

The object $\sfrac{1}{2}\beta$ defined in \eqref{inverse} is the gauge-fixed form of $\eqref{fullbeta}$. Then, by \ref{betalemma}, it defines a $\beta\in Z^1(G_b,\ZZ/2)$.

It remains to show that these maps are indeed inverses. Since $\sfrac{1}{2}\beta$ is the image of a $\ZZ/2$-valued cocycle $\beta$, $\omega$ can be written with a minus sign like $\omega=\alpha-\sfrac{1}{2}\beta\cup t$. Note also that $s(\bar g)=p^{t(g)}g$. Then
\begin{equation}
\omega(g,h)=\alpha(\bar g,\bar h)-\sfrac{1}{2}\beta(\bar g)t(h)=\omega(p^{t(g)}g,p^{t(h)}h)-\omega(p^{t(g)}g,p)t(h)=\omega(g,h),
\end{equation}
\begin{equation}
\alpha(g_b,h_b)=\omega(s(g_b),s(h_b))=\alpha(g_b,h_b)+\sfrac{1}{2}\beta(g_b)t(s(g_b))=\alpha(g_b,h_b),
\end{equation}
\begin{equation}
\sfrac{1}{2}\beta(g_b)=\omega(s(g_b),p)=\alpha(g_b,1)+\sfrac{1}{2}\beta(g_b)t(p)=\sfrac{1}{2}\beta(g_b).
\end{equation}
\end{proof}

\begin{theorem}
$H^2(\cG, U(1)_T)$ equals, as a set, the set of pairs $(\alpha,\beta)$ (see \ref{isomlemma}) modulo the equivalence $(\alpha',\beta) \sim (\alpha,\beta)$ if $\alpha' = \alpha + \delta_T \lambda$ with $\lambda$ a cochain in $C^1(\cG,U(1)_T)$ satisfying $\lambda(s(g_b)p) = \lambda(s(g_b)) + \lambda(p)$.\footnote{Had we chosen a different representative $\rho' = \rho + \delta \mu$ of $[\rho]$ to describe the extension of $G_b$ by $\ZZ_2^F$, we would have considered a different set of cochains $\alpha$ (modulo coboundaries), shifted by $ \sfrac{1}{2} \beta \cup \mu$, but their counting would be the same.}
\end{theorem}

\begin{proof}
The preceding lemmas show that the set of twisted cocycles $\omega$ satisfying the gauge conditions \eqref{gauge1} and \eqref{gauge2} is equivalent to the set of pairs $(\alpha, \beta)$. After transforming $\omega$ into this gauge, there remains freedom to choose $L(g)$ for each $g\in\cG$ such that $t(g)=0$, and to shift $L(p)$ by $\sfrac{1}{2}$. We have already seen that $\beta$ is invariant under an arbitrary gauge transformation. However, there is some residual gauge freedom for $\alpha$.

Let $\omega' = \omega + \delta_T \lambda$ be another $2$-cocycle satisfying the gauge conditions. It takes the form $W - \delta_T L'$, with $L'$ possibly differing from $L$ in its values on $s(g_b)$ and $p$. We see from $\delta_T \lambda = \omega' - \omega = \delta_T (L'-L)$ that $\lambda = L'-L + \kappa$ where $\kappa$ is a twisted $1$-cocycle. Then, by \eqref{gaugeL}, $\lambda(s(g_b)p)=\lambda(s(g_b)) + \lambda(p)$. The quantities $L(p)$, $L'(p)$, $\kappa(p)$, and therefore $\lambda(p)$, can each be chosen to be $0$ or $\sfrac{1}{2}$. Finally, by \eqref{inverse}, this freedom in gauge-fixed $\omega$ translates into the desired freedom in $\alpha$.

\end{proof}


\end{document}